\documentclass[11pt]{article}
\usepackage{setspace}
\usepackage[papersize={8.5in,11in},top=1.5in,left=1in,right=1in,bottom=1.5in]{geometry}
\usepackage{fullpage}
\usepackage{booktabs} 
\usepackage{url}
\usepackage{xcolor,color,url,eurosym,graphicx,amssymb,amsfonts,amsmath,algorithm,algorithmic,bbm,xspace}
\usepackage[most]{tcolorbox}
 \usepackage{empheq}
 \usepackage{hyperref}

\DeclareMathOperator{\diag}{Diag}

\newtcbox{\mymath}[1][]{%
    nobeforeafter, math upper, tcbox raise base,
    enhanced, colframe=blue!30!black,
    colback=blue!30, boxrule=1pt,
    #1}

\newcommand{\copynodes}{{\ensuremath{\overline{V}}}}
\newcommand{\copyedges}{{\ensuremath{\overline{E}}}}

\newcommand{\beql}[1]{\begin{equation}\label{#1}}
\newcommand{\eeq}{\end{equation}}

\newcommand{\field}[1]{\mathbb{#1}} 

\newcommand{\NPhard}{{\ensuremath{\mathbf{NP}}-hard}\xspace}

\newcommand{\spara}[1]{\smallskip\noindent{\bf #1}}

\newcommand{\hide}[1]{}
\newcommand{\squishlist}{
 \begin{list}{$\bullet$}
  {  \setlength{\itemsep}{0pt}
     \setlength{\parsep}{3pt}
     \setlength{\topsep}{3pt}
     \setlength{\partopsep}{0pt}
     \setlength{\leftmargin}{2em}
     \setlength{\labelwidth}{1.5em}
     \setlength{\labelsep}{0.5em}
} }
\newcommand{\squishlisttight}{
 \begin{list}{$\bullet$}
  { \setlength{\itemsep}{0pt}
    \setlength{\parsep}{0pt}
    \setlength{\topsep}{0pt}
    \setlength{\partopsep}{0pt}
    \setlength{\leftmargin}{2em}
    \setlength{\labelwidth}{1.5em}
    \setlength{\labelsep}{0.5em}
} }

\newcommand{\squishdesc}{
 \begin{list}{}
  {  \setlength{\itemsep}{0pt}
     \setlength{\parsep}{3pt}
     \setlength{\topsep}{3pt}
     \setlength{\partopsep}{0pt}
     \setlength{\leftmargin}{1em}
     \setlength{\labelwidth}{1.5em}
     \setlength{\labelsep}{0.5em}
} }

\newcommand{\squishend}{
  \end{list}
}

\newcommand{\squishlistt}{
 \begin{list}{---}
  {  \setlength{\itemsep}{0pt}
     \setlength{\parsep}{3pt}
     \setlength{\topsep}{3pt}
     \setlength{\partopsep}{0pt}
     \setlength{\leftmargin}{2em}
     \setlength{\labelwidth}{1.5em}
     \setlength{\labelsep}{0.5em}
} }

\usepackage{fancyhdr}   
\usepackage{tikz}
\usepackage{geometry}
\usepackage{graphicx} 
\usepackage{amsthm}
\usepackage{microtype}
\usepackage{authblk}
\usepackage{algorithmic}
\usepackage{algorithm}

\newtheorem{lm}{Lemma}
\newtheorem{theorem}[lm]{Theorem}
\theoremstyle{definition}
\newtheorem{defi}[lm]{Definition}

\newtheorem{cor}[lm]{Corollary}
\newtheorem{qq}[lm]{Problem}
\newtheorem{eg}[lm]{Example}

\renewcommand{\(}{\left(}
\renewcommand{\)}{\right)}

\newcommand\myeq{\mathrel{\stackrel{\makebox[0pt]{\mbox{\normalfont\tiny def}}}{=}}}

\title{Opinion Dynamics with Varying Susceptibility to Persuasion }
\author[1]{Rediet Abebe}
\author[1]{Jon Kleinberg}
\author[2]{David Parkes}
\author[3]{Charalampos E. Tsourakakis}
\affil[1]{Cornell University}
\affil[2]{Harvard University}
\affil[3]{Boston University}
\date{}                     
\begin{document}

\maketitle

\newcommand{\omt}[1]{}
\newcommand{\xhdr}[1]{\vspace{1.7mm}\noindent{{\bf #1.}}}

\begin{abstract}
A long line of work in social psychology has studied variations in people's susceptibility to persuasion -- the extent to which they are willing to modify their opinions on a topic.  This body of literature suggests an interesting perspective on theoretical models of opinion formation by interacting parties in a network: in addition to considering interventions that directly modify people's intrinsic opinions, it is also natural to consider interventions that modify people's susceptibility to persuasion.  In this work, we adopt a popular model for social opinion dynamics, and we formalize the {\em opinion maximization} and {\em minimization} problems where interventions happen at the level of susceptibility. Specifically, we ask: 

\begin{quotation}
Given an opinion dynamics model, and a set of agents, each of whom has an innate opinion that reflects the agent's intrinsic position on a topic and a susceptibility-to-persuasion parameter measuring the agent's propensity for changing his or her opinion, how should we modify the agents' susceptibilities in order to maximize (or minimize) the total sum of opinions at equilibrium? 
\end{quotation}

We study both the version when there is not a budget on the number of agents that we can target, and when there is. For the former, we provide a polynomial time algorithm for both the {\em maximization} and the {\em minimization}  formulations. In contrast, we show that budgeted version is \NPhard, but unlike prior work focusing on opinion maximization, our objective function is neither sub- nor super-modular in the target set. For this version, we propose a greedy heuristic that we compare experimentally against natural baseline methods.  On the experimental side, we perform an empirical study of our proposed methods on   several datasets, including a Twitter network with real opinion estimates. Our experimental findings show that our proposed tools outperform natural baselines and can achieve a multi-fold effect on the total sum of opinions in the network. 
\end{abstract}

\section{Introduction}
A rich line of empirical work in development and social psychology has studied people's susceptibility to persuasion. This property measures the extent to which individuals are willing to modify their opinions in reaction to the opinions expressed by those around them,
and it is distinct from the opinions they express. Research in the area has ranged from adolescent susceptibility to peer pressure related to risky and antisocial behavior \cite{peer2,peer1,evans1992measuring,peer4,peer3,health} to the role of susceptibility to persuasion in politics \cite{pol3,pol1,pol2}. Individuals' susceptibility to persuasion 
can be affected by specific strategies and framings aimed at increasing susceptibility \cite{cialdini1,cialdini2,crowley,hci,mcguire,socpsy2,socpsy1}. For 
instance, if it is known that an individual is receptive to \emph{persuasion by authority}, one can adopt a strategy that utilizes arguments from official sources and authority figures 
to increase that individuals' susceptibility to persuasion with respect to a particular topic.  

Modifying network opinions has far-reaching implications including product marketing, public health campaigns, the success of political candidates, and public opinions on issues of global 
interest. In recent years, there has also been work in Human Computer Interaction focusing on 
{\em persuasive technologies}, which are designed with the goal of changing a person's attitude or behavior \cite{persuasivetech1,persuasivetech2,hci}. This work has shown that not only do people differ in their susceptibility to persuasion, but that persuasive technologies can also be adapted 
to each individual to change their susceptibility to persuasion. 

Despite the long line of empirical work emphasizing the importance of
individuals' susceptibility to persuasion, to our knowledge
theoretical studies of opinion formation models have not focused on
interventions at the level of susceptibility.  
Social influence studies have considered interventions that
directly act on the opinions themselves, both in discrete models
(as in the work of Domingos and Richardson \cite{dr}, Kempe et al.
\cite{kkt}, and related applications \cite{aak,berger2007consumers,hb,turf})
and more recently in a model with continuous opinions
(in the work of Gionis et al.  \cite{gtt}).

In this work, we adopt an opinion formation model inspired by the work of DeGroot \cite{de} and Friedkin and Johnsen \cite{fj}, and we initiate a study of the impact of interventions at the level of susceptibility. In this model, each agent $i$ is endowed with an {\em innate opinion} $s_i \in [0,1]$ and a parameter representing susceptibility to persuasion, which we will call 
the {\em resistance parameter}, $\alpha_i \in (0,1]$. 
The innate opinion $s_i$ reflects the intrinsic position of agent $i$
on a certain topic, while  $\alpha_i$ reflects the agent's willingness, 
or lack thereof, to conform with the opinions of neighbors in the social 
network.
We term $\alpha_i$ the agent's ``resistance'' because a high value of
$\alpha_i$ corresponds to a lower tendency 
to conform with neighboring opinions.
According to the opinion dynamics model, the final opinion of each
agent $i$ is a function of the social network, the set of innate
opinions, and the resistance parameters, determined by computing
the {\em equilibrium} state of a dynamic process of opinion updating.
For more details on the model,
see Section~\ref{sec:model}.  We study the following natural question:

\begin{tcolorbox}
\begin{qq}
\label{prob1} 
Given an opinion dynamics model, and a set of agents, each of whom has
an innate opinion that reflects the agent's intrinsic position on a topic, and
a resistance parameter measuring the agent's propensity for changing his or her
opinion, how should we change the agents' susceptibility in order to
maximize (or minimize) the total sum of opinions at equilibrium?
\end{qq}
\end{tcolorbox}

We call the set of agents whose susceptibility parameters we change
the \emph{target-set}. In most influence maximization or opinion
optimization settings, the case where there are no constraints on the
size of the target-set is generally not interesting since it often involves
simple solutions that modify a parameter to a specific value
across essentially all agents.
This is not the case in our setting -- the unbudgeted version
remains technically interesting, and the optimal solution does not
generally involve targeting all agents. We will therefore consider
both the unbudgeted and the budgeted versions of this problem.

\spara{The Present Work.} Within the opinion dynamics model we study, 
we formalize the problem of optimizing the total sum of opinions at
equilibrium by  modifying the agents' susceptibility to persuasion.
We formalize both the
maximization and minimization versions of the problems, 
as well as an unbudgeted, and a budgeted variation of these problems.

For the unbudgeted maximization and minimization problems we provide efficient polynomial time algorithms. In contrast to the unbudgeted problem, the budgeted problem is \NPhard. We show that our objective is neither sub- nor super-modular. We prove 
that there always exists an optimal solution where the resistance parameters
$\alpha_i$ are set to extreme values (see Theorem~\ref{thm:esv}). We use this theorem as our guide to provide a heuristic that greedily optimizes the 
agents' susceptibility to persuasion.  


We evaluate our unbudgeted and budgeted methods on real world network
topologies, and on both synthetic and real opinion values. In
particular, we use a Twitter network and opinion data on the Delhi
legislative  assembly elections based on work of 
De et al. \cite{de2014learning}
to see how our methods can potentially
modify the final outcome of opinion dynamics. Our findings 
indicate that intervening at the susceptibility level can
significantly affect the total sum of opinions in the network.

\section{Model}
\label{sec:model}

Let $G = (V,E)$ be a simple, undirected graph, where $V = [n]$ is the set of agents and $E$ is the set of edges. Each agent $i \in V$ is associated with an \emph{innate opinion} $s_i \in [0,1]$, where higher values correspond to more favorable opinions towards a given topic and a parameter measuring an agent's susceptibility to persuasion $\alpha_i \in (0,1]$, where higher values signify agents who are less susceptible to changing their opinion. We call $\alpha_i$ the \emph{resistance parameter}. The opinion dynamics evolve in discrete time according to the following model, inspired by the work of \cite{de,fj}: 

\begin{equation}
\label{eq:dynamics}
x_i(t+1) = \alpha_i s_i + (1-\alpha_i) \frac{ \sum\limits_{j \in N(i)}  x_j(t) }{\deg(i)}.
\end{equation}

\noindent   Here, $N(i) = \{ j : (i, j) \in E \}$ is the set of neighbors of $i$, and $\deg(i)=|N(i)|$. It is known that this dynamics converges to a unique equilibrium if  $\alpha_i > 0$ for all $i \in V$ \cite{deb}. The equilibrium opinion vector $z$  is the solution to a linear system of equations: 

\begin{equation} 
z = (I - (I-A)P)^{-1} A s,
\end{equation}

\noindent where $A=\diag{(\alpha)})$ is a diagonal matrix where entry $A_{i, i}$ corresponds to $\alpha_i$ and $P$ is the random walk matrix, i.e., $P_{i, j} = \frac{1}{deg(i)}$ for each edge $(i, j)\in E(G)$, and zero otherwise. We call $z_i$ the \emph{expressed opinion} of agent $i$.

\begin{defi}
Given the opinion dynamics model \eqref{eq:dynamics}, a social network $G = (V,E)$, and the set of parameters $\{s_i\}_{i \in V},\{\alpha_i\}_{i \in V}$, we define the {\em total sum of opinions} at the equilibrium as 

$$ f(s, \alpha) \myeq  \vec{1}^T z =    \vec{1}^T(I - (I-A)P)^{-1} A s.$$
\end{defi}

Our goal is to optimize the objective $f(s, \alpha)$, by choosing certain $\alpha_i$ parameters in a range $[\ell, u]$,  where $0 < \ell \leq u$.  We consider the following problems:

\begin{tcolorbox}
\begin{qq}[Opinion Maximization]
\label{probmax} 
Given a positive integer $k$, find a set of nodes $S \subseteq V, |S| \leq k$, and  set the resistance parameters $\alpha_i \in [\ell, u]$ for each $i \in S$, in order to maximize the total sum of opinions  $f(s, \alpha)$. We refer to the case $k\geq n$ as the {\em unbudgeted opinion maximization problem} and the case $k < n$ as the {\em budgeted opinion maximization problem}.
\end{qq}
\end{tcolorbox}

Similarly, we define the minimization version. 
  
\begin{tcolorbox}
\begin{qq}[Opinion Minimization]
\label{probmin} 
Given a positive integer $k>0$, find a set of nodes $S \subseteq V, |S| \leq k$, and  set the resistance parametes $\alpha_i \in [\ell, u]$ for each $i \in S$, in order to minimize the total sum of opinions $f(s, \alpha)$. The unbudgeted and budgeted versions are defined as above. 
\end{qq}
\end{tcolorbox}

Previous work has focused on modifying the agents' innate or expressed opinions rather than their resistance parameters. We observe that targeting agents at the level of susceptibility versus at the level of innate opinions can lead to very different results, as the next example shows.

\begin{eg}
\label{eg:motivation}
Consider the following network on $n$ nodes, where each node is annotated with the values $(s_i, \alpha_i)$, and consider interventions at the level of
either innate opinions or resistance parameters with the goal of
maximizing the sum of expressed opinions.

\begin{figure}[!ht]
\centering
\begin{tikzpicture}
[scale=1,auto=left,every node/.style={circle,fill=gray!20}]
  \node[fill=blue!50] (n1) at (0, 0) {\tiny{(0,$\epsilon$)}};
  \node[fill=blue!50] (n2) at (1.5, 0) {\tiny{(0,$\epsilon$)}};
  \node[fill=blue!50] (n3) at (3.5, 0) {\tiny{(0,$\epsilon$)}};
  \node[fill=blue!50] (n4) at (5, 0) {\tiny{(0,$\epsilon$)}};
  \node[fill=blue!50] (n5) at (2.5, 1.5) {\small{(1,$\frac{\epsilon}{c}$)}};
  \foreach \from/\to in {n1/n5, n2/n5, n3/n5, n4/n5}
    \draw (\from) -- (\to);
\draw[dotted](n2) -- (n3); 
\end{tikzpicture}
\end{figure}

Suppose $\epsilon$ is a value arbitrarily close to 0. As we increase
$c$, the total sum of opinions in the network  tends to 0. If we can
target a single agent, then an optimal  intervention at the 
level of innate opinions
changes the innate opinions of one of the leaves to 1. At
equilibrium, the total sum of opinions (for small $\epsilon$ and 
large $c$) is close to $\frac{n}{n-1} \approx 1$.
However, if we can target at the level of susceptibility, then we
change the central node's $\alpha_i$ from $\frac{\epsilon}{c}$ to 1.
Then, at equilibrium the total sum of opinions becomes
$f(s,\alpha) = n (1 - \phi(\epsilon,c))$ for a function $\phi()$
going to $0$ as $\epsilon$ goes to $0$ and $c$ goes to infinity. 
The improvement is $O(n)$, and hence very large compared
to the intervention at the level of innate opinions.

We can construct an instance on the same graph structure to
show the contrast in the opposite direction as well,
where targeting at the level of innate
opinions cause a very large improvement compared to targeting at the
level of susceptibility.  Starting with the example in the figure,
change the innate
opinion of the center node to be $0$ and the resistance parameter to
be $1$. In equilibrium, the sum of expressed opinions is $0$. If we
can target agents at the level of innate opinions, then changing the
innate opinion of the center node to be $1$ will lead to a sum of
expressed opinions of $n$, whereas we cannot obtain a sum of expressed
opinions greater than $0$ by targeting at the level of susceptibility.

\end{eg}

As a first general result about our model, 
we prove that there always exists an optimal solution that involves
extreme values for the resistance parameters $\alpha_i$.

\begin{theorem}
\label{thm:esv}
There always exists an optimal solution for the opinion optimization problems where $\alpha_i=\ell$ or $\alpha_i=u$. 
\end{theorem} 

\begin{proof} 
Our proof relies on a random-walk interpretation of opinion dynamics \cite{ghaderi2013opinion,gionis2013opinion}. Before we give our proof, we describe 
this connection between absorbing random walks and opinion dynamics for completeness.  Gionis et al. \cite{gionis2013opinion} consider a random walk with absorbing states on an augmented graph $H=(X,R)$ that is constructed as follows: 
\begin{itemize}
\item[(i)]
The set of vertices $X$ of $H$ is defined as $X = V\cup {\copynodes}$,
where ${\copynodes}$ is a set of $n$ new nodes such that for each node $i \in V$ there is a copy $\sigma(i)\in\copynodes$;
\item[(ii)]
The set of edges $R$ of $H$ includes all the edges $E$ of $G$, plus a new set of edges between each node $i\in V$ and its copy $\sigma(i)\in\copynodes$.
That is, $R=E\cup {\copyedges}$, and ${\copyedges}=\{ (i,\sigma(i)) \mid  i\in V \}$;
\item[(iii)]
Finally, the weights of all edges in $R$ are set to 1, assuming that $G$ is unweighted. 
\end{itemize}

\noindent A key observation in \cite{gionis2013opinion} is that the
opinion vector at the equilibrium can be computed by performing the following
absorbing random walk on the graph $H$.  To determine the expressed opinion
of node $i \in V$, we start a random walk from $i$; with probability
$1 - \alpha_i$ we follow an incident edge chosen uniformly at random 
from $E$, and with probability $\alpha_i$ we follow the edge
from $i$ to $\sigma(i)$, ending the walk.
If the walk continues to a node $j$, we perform the same type of step ---
chosing an edge of $E$ incident to $j$ with probability $1 - \alpha_j$,
and moving to $\sigma(j)$ with probability $\alpha_j$ --- and iterate.
It is shown in \cite{gionis2013opinion} that this walk represents 
the opinion dynamics process in the following sense: the 
expressed opinion $z_i^*$ of node $i$ in the equilibrium of the 
process is equal to the expected value of the innate opinion
$s_j$ over all endpoints $\sigma(j)$ of the walk starting from $i$.

This implies that the equilibrium
opinion $z_i^*$ of node $i$ is given by the following general form:

\begin{equation} 
\label{eq:gionis}
z_i^* = \alpha_i s_i + (1-\alpha_i) p_{ii} z_i^* +  (1-\alpha_i) (1-p_{ii})y_i,
\end{equation}

\noindent where $p_{ii}$ is the probability that the random walk that starts at $i$ returns again to $i$ before absorption, and $y_i$ is the expected opinion if the random walk gets absorbed to a node $\sigma(j) \neq \sigma(i)$.

 We assume, without loss of generality, that the 
opinion optimization problem at hand is the opinion 
maximization problem, although the proof can easily be adopted 
to the opinion minimization case. 
Suppose we have a solution for the opinion maximization problem such
that $\alpha_i$ lies strictly between $\ell_i$ and $u_i$.
We wish to show that setting $\alpha_i$
to be either $\ell_i$ or $u_i$ will yield a sum of expressed opinions that
is at least as large. Proceeding in this way one node at a time,
we will get a solution in which all $\alpha_i$ are either $\ell_i$ or $u_i$.

By solving Equation~\eqref{eq:gionis} for
$z_i^*$ we obtain,
\begin{equation} 
\label{eq:combination}
z_i^*  = \frac{\alpha_i}{ 1-(1-\alpha_i)p_{ii}} s_i + \frac{(1-\alpha_i)(1-p_{ii})}{ 1-(1-\alpha_i)p_{ii}} y_i. 
\end{equation}
\noindent Therefore $z_i^*$  is the convex combination of two
probabilities $s_i, y_i$. Since we wish to maximize $z_i^*$,  if $s_i
\leq y_i$ we need to minimize $\dfrac{\alpha_i}{ 1-(1-\alpha_i)p_{ii}}
$; otherwise we will minimize $\dfrac{(1-\alpha_i)(1-p_{ii})}{
1-(1-\alpha_i)p_{ii}}$. Note that the function $g(x) =
\dfrac{x}{1-(1-x)p_{ii}}$ for $x \in [\ell_i,u_i]$ is monotone increasing,
and therefore maximized at $x = u_i$ and minimized at $x=\ell_i$.
Therefore, optimality is obtained by setting $\alpha_i$ to be either
$\ell_i$ or $u_i$. This yields a sum of expressed opinions that is at
least as large as the original solution.
\end{proof}

\section{Unbudgeted Opinion Optimization}
\label{sec:unbudgeted}

In this section, we show that when there are no constraints on the size of the target-set, both the opinion maximization and minimization problems can be solved in polynomial time. 

\subsection{Unbudgeted Opinion Minimization}

Recall from Section~\ref{sec:model}, $A=\diag{(\alpha)})$ is a diagonal matrix where entry $A_{i, i}$ corresponds to $\alpha_i$ and $P$ is the random walk matrix. 

We formalize the unbudgeted opinion minimization problem as follows. 

\begin{align}
\label{eq:opmin1}
\begin{array}{ll@{}ll}
\text{minimize}  &  \vec{1}^T z &    &  \\
\text{subject to} &z = (I - (I-A)P)^{-1} A s & & \\
                 &                A \succeq \ell I                                &   & \\  
                 &               A \preceq u I                                 & & \\ 
\end{array}
\end{align}

 We show that this problem can be solved in polynomial time.

\begin{theorem} 
\label{thm:main}
The {\em Opinion Minimization} formulation~(\ref{eq:opmin1}) is solvable in polynomial time.
\end{theorem}

\noindent Before we proceed with the proof of Theorem~\ref{thm:main}, we rewrite 
the objective of formulation~(\ref{eq:opmin1}) in a way that is convenient for algebraic 
manipulation.  Define $X = A^{-1}$. Then,

\begin{align*}
  (I - (I-A)P)^{-1} A &=  (I - (I-A)P)^{-1} (A^{-1})^{-1} =  (I - (I-A)P)^{-1} X^{-1} \\ 
  &= (X - (X-XA)P)^{-1} = (X-(X-I)P)^{-1}.  
\end{align*}

Therefore, we rewrite our original formulation to optimize over  the set of diagonal matrices $X$ whose entries lie in $[ {\frac{1}{u},\frac{1}{\ell}} ]$.

\begin{align}
\label{eq:opmin}
\begin{array}{ll@{}ll}
\text{minimize}  & \displaystyle \vec{1}^T   (X-(X-I)P)^{-1} s &  \\
\text{subject to}        &                                                X \succeq \frac{1}{u} I  & \\  
                 &                                                X \preceq \frac{1}{\ell} I& \\ 
                 &                                                X \text{~~diagonal} & \\ 
\end{array} 
 \end{align}

\begin{lm}
The following set is convex:  
$$\mathcal{K} = \{ X \in \field{R}^{n\times n}:X \text{~diagonal}, \frac{1}{u}\leq x_{ii} \leq \frac{1}{\ell} \text{~for all~} i \in [n] \}.$$
\end{lm}

\begin{proof} 
Let $X_1,X_2 \in \mathcal{K}$, $\lambda \in (0,1)$. Matrix $Y = \lambda X_1 + (1-\lambda)X_2$ is clearly diagonal, and any diagonal entry $y_{ii}= \lambda x_{ii}^{1} + (1-  \lambda) x_{ii}^{2}$ satisfies,
$$\frac{1}{u}\leq y_{ii} \leq \frac{1}{\ell}.$$
Therefore, $Y \in \mathcal{K}$, which implies the claim.
\end{proof}

\begin{lm}
\label{psd}
All eigenvalues of matrix  $Z = (X - (X-I)P), X \in \mathcal{K}$ are strictly positive reals (in general $Z$ is not  symmetric). Furthermore, $Z$ is positive definite, i.e., for any $y \in \field{R}^n$ the quadratic form $y^T Z y$ is non-negative.
\end{lm}

\noindent Notice that the above two statements are not simultaneously true for non-symmetric matrices: a non-symmetric matrix with all eigenvalues positive is not necessarily positive definite.  On the other hand, if a non-symmetric matrix is positive definite, all its eigenvalues have to be positive.   
 
 \begin{proof}  
 First we prove that all eigenvalues are real. We provide a sequence of similarity transformations that prove $Z = (X-(X-I)P)$ is similar to a symmetric matrix.  To simplify notation, when $U = R^{-1} A R$, where $R$ is a diagonal matrix,  we write $U \equiv A$.

 \begin{align*}
X - (X-I)P &\equiv X-D^{+\tfrac{1}{2}} \( {  (X-I) D^{-1}A } \) D^{-\frac{1}{2}}  \\ 
&= X-  (X-I) D^{-\frac{1}{2}}AD^{-\frac{1}{2}}   
   \\
&\equiv X-  (X-I)^{-\frac{1}{2}} (X-I) D^{-\frac{1}{2}}AD^{-\frac{1}{2}}   (X-I)^{+\frac{1}{2}}   \\
&=X-  (X-I)^{+\frac{1}{2}}   D^{-\frac{1}{2}}AD^{-\frac{1}{2}}   (X-I)^{+\frac{1}{2}} 
 \end{align*}
 
 \noindent Since the adjacency matrix $A$ is symmetric, the latter matrix is also symmetric, and therefore has real eigenvalues.
 
 To show all eigenvalues are positive, we invoke the Gershgorin circle theorem.  Notice that $Z_{ii}=x_{ii}$ (under our assumption that the graph is simple) and the sum of non-diagonal entries in row $i$ is $-(x_{ii}-1)$. Therefore, every eigenvalue of $Z$ lies in at least one of the Gershgorin discs:
\begin{align*}
 |\lambda - x_{ii}| \leq (x_{ii}-1) &\Leftrightarrow   1 \leq \lambda \leq 2x_i -1.
\end{align*}
 
We now show  that the quadratic form associated with $Z$ is non-negative.  Recall first     that $x_{ii} \geq 1$ and $\deg(i)\geq 1$, therefore  $$ \frac{1}{\deg(i)}+ \frac{1}{\deg(j)} \geq 2.$$   

 \noindent For any $y \in \field{R}^n$,

\begin{eqnarray*}
y^T Z y &=& y^T (X - (X-I)P)  y = y^T X y - y^T  (X-I)y  \\  
& =& \sum\limits_{i=1}^n x_{ii} y_i^2 - \sum\limits_{i=1}^n (x_{ii}-1) \frac{y_i}{\deg(i)} \sum\limits_{j \sim i} y_j\\
& =& \sum\limits_{i=1}^n  y_i^2 + \sum\limits_{i=1}^n (x_{ii}-1) \( { y_i^2 - \frac{y_i}{\deg(i)} \sum\limits_{j \sim i} y_j } \) \\   
& =& \sum\limits_{i=1}^n  y_i^2 + \sum\limits_{(i,j) \in E(G)}  (x_{ii}-1) \( { y_i^2 +y_j^2 - (\frac{y_iy_j}{\deg(i)}+ \frac{y_iy_j}{\deg(j)} ) } \)  \\
& \geq& \sum\limits_{i=1}^n  y_i^2 + \sum\limits_{(i,j) \in E(G), y_iy_j \geq 0}  (x_{ii}-1) (y_i^2+y_j^2-2y_iy_j) + \\
 & & \sum\limits_{(i,j) \in E(G), y_iy_j < 0}  (x_{ii}-1)    \( { y_i^2 +y_j^2 + (\frac{|y_iy_j|}{\deg(i)}+ \frac{|y_iy_j|}{\deg(j)} ) } \)\\  
 & = & \sum\limits_{i=1}^n  y_i^2 + \sum\limits_{(i,j) \in E(G), y_iy_j \geq 0}  (x_{ii}-1) (y_i-y_j)^2 + \\
  & & \sum\limits_{(i,j) \in E(G), y_iy_j < 0}  (x_{ii}-1)    \( { y_i^2 +y_j^2 + (\frac{|y_iy_j|}{\deg(i)}+ \frac{|y_iy_j|}{\deg(j)} ) } \) \\  
& \geq & 0
\end{eqnarray*}
 \end{proof}

\noindent We proceed to prove that the objective is convex.  

\begin{theorem} 
The matrix $(X-(X-I)P)^{-1}$ is a convex function of $X \in \mathcal{K}$. 
\end{theorem}

\begin{proof}
Let $X_1,X_2 \in \mathcal{K}$, and let $Z_i:=(X_i-(X_i-I)P)$ for $i=1,2$.  Also, fix $\lambda \in (0,1)$, and let $Z(\lambda) = \lambda Z_1 + (1-\lambda)Z_2$. We prove the following statement, that implies that $(X-(X-I)P)^{-1}$:

$$ \lambda Z_1^{-1}+ (1-\lambda) Z_2^{-1}   \succeq (\lambda Z_1 + (1-\lambda)Z_2)^{-1}.$$

\noindent Since $Z_1,Z_2$ are positive definite, we obtain that for any $y \in \field{R}^n$  

$$y^TZ(\lambda)y = \lambda    y^T Z_1 y + (1-\lambda) y^T Z_2 y \geq 0,$$  

\noindent so  $Z(\lambda)$ is positive definite, and so is $Z(\lambda)^{-1}$, 

$$ Z Z^{-1} = I \Rightarrow \frac{dZ}{d\lambda}      Z^{-1} + Z  \frac{dZ^{-1}}{d\lambda} = 0  \Rightarrow   \frac{dZ^{-1}}{d\lambda}  = -Z^{-1}\frac{dZ}{d\lambda}      Z^{-1}.$$
 
\noindent By differentiating twice $Z(\lambda)$ with respect to $\lambda$, we also obtain 
$ \frac{d^2Z}{d\lambda^2}=0$.  Therefore, the second derivative of the inverse is given by,

$$  \frac{d^2Z^{-1}}{d\lambda^2}  =
- \frac{d Z^{-1}}{d\lambda} \frac{dZ}{d\lambda} Z^{-1} - Z^{-1}\frac{dZ}{d\lambda}\frac{d Z^{-1}}{d\lambda}  = 
2Z^{-1}\frac{dZ}{d\lambda}Z^{-1}\frac{dZ}{d\lambda}Z^{-1}.$$

\noindent Pick any non-zero vector $u$  and consider 
following pair of vector/matrix valued functions:

$$ v(\lambda) = \frac{dZ}{d\lambda}\frac{dZ^{-1}}{d\lambda} u\quad\quad\text{ and }\quad\quad
\varphi(\lambda) = u^T Z^{-1}(\lambda) u.$$

\noindent The second derivative of $\varphi(\lambda)$  satisfies the following inequality:

$$ \varphi(\lambda) = u^T  \frac{d^2Z^{-1}}{d\lambda^2} u = 2 v^T(\lambda) Z^{-1}(\lambda) v(\lambda) \ge 0,$$

\noindent since $Z^{-1}$ is  positive definite. Therefore $\varphi(\lambda)$ is a convex function of $\lambda$ for $\lambda \in (0,1)$.  Thus, for any such $\lambda$

\begin{align}&(1-\lambda)\varphi(0) + \lambda\varphi(1) - \varphi(\lambda) \ge 0\\
&\iff u^T \left[ (1-\lambda) X^{-1} + \lambda Y^{-1} - ((1-\lambda) X + \lambda Y)^{-1}\right] u \ge 0. 
\end{align}

\noindent This implies  that $ \lambda Z_1^{-1} + (1-\lambda) Z_2^{-1}  - (\lambda Z_1+(1-\lambda) Z_2  )^{-1}$ is positive semidefinite, and thus our claim that $(X-(X-I)P)^{-1}$ is a convex function holds. 
\end{proof}

\noindent We can prove that $\Big((X-(X-I)P)^{-1}\Big)^T$ is also a convex function in a similar way. We therefore obtain the following corollary: 

\begin{cor} 
\label{maincor}
For any fixed $y \in \field{R}^n$, we can optimize  the quadratic forms,

$$ y^T(X-(X-I)P)^{-1}y, \text{~~~} y^T\Big( (X-(X-I)P)^{-1} + \big( y^T\Big( (X-(X-I)P)^{-1}\big)^T \Big)y$$

\noindent  over all diagonal matrices $X \in \mathcal{K}$  in polynomial time.  
\end{cor}

\noindent We can use off-the-shelf interior point methods for convex optimization, see \cite[Chapter 11]{boydbook}. To conclude the proof of Theorem~\ref{thm:main}, we need to prove that our objective can be brought in the quadratic form of Corollary~\ref{maincor}. 

\begin{proof}[Proof of Theorem~\ref{thm:main}] 
For simplicity, let $Z=(X-(X-I)P)$. We consider the quadratic,

\begin{align*}
(s+ \vec{1})^T   Z^{-1} (s+ \vec{1}) &= s^TZ^{-1} s + \vec{1}^T Z^{-1} \vec{1}+ 
 \vec{1}^T Z^{-1} s + 
s^T Z^{-1} \vec{1} \\ 
 &=s^TZ^{-1} s + n + \vec{1}^T Z^{-1} s + s^T \vec{1} \Leftrightarrow\\ 
 \vec{1}^T Z^{-1} s &= (s+ \vec{1})^T   Z^{-1} (s+ \vec{1}) - s^TZ^{-1} s  - C \\ 
 &=  \langle \underbrace{  (s+ \vec{1}) (s+ \vec{1})^T   - ss^T}_{Q}, Z^{-1} \rangle - C, 
\end{align*}
 
\noindent where $C=n+\sum\limits_{i=1}^n s_i$ is a constant that does not depend on the variable diagonal matrix $X$.  Furthermore, $Q$ is symmetric, i.e., $Q=Q^T$. Let  $Q=\sum_{i=1}^n\nu_i q_i q_i^T$ be the eigendecompositon of $Q$. We observe that 

\begin{align*}
\textrm{trace}(Z^{-1}Q) &= \textrm{trace}\left(Z^{-1}\left(\sum_{i=1}^n\nu_i q_i q_i^T\right)\right) = \textrm{trace}\left(\sum_{i=1}^n\nu_iZ^{-1} q_i q_i^T\right)\\
&= \sum_{i=1}^n\nu_i \textrm{trace}\left( Z^{-1} q_i q_i^T\right) = \sum_{i=1}^n\nu_i \textrm{trace}\left(q_i^T Z^{-1} q_i \right)\\
&= \sum_{i=1}^n\nu_i \left(q_i^T Z^{-1} q_i \right) \tag{*} =  \sum_{i=1}^n\nu_i \left(q_i^T \frac{1}{2}\big( Z^{-1}+(Z^{-1})^T\big) q_i \right) \\ 
&= \frac{1}{2} \textrm{trace}\left( \big( Z^{-1}+(Z^{-1})^T\big) Q\right) 
\end{align*}
 
 \noindent In (*), we used the fact that for any square matrix $A$ 
 
$$ x^TAx = \frac{1}{2}x^T(A+ A^T)x.$$  

\noindent Therefore, we conclude that our objective is equivalent to minimizing 
 $$ \frac{1}{2} \langle  (s+ \vec{1}) (s+ \vec{1})^T   - ss^T, Z^{-1}+(Z^{-1})^T \rangle,$$

 \noindent  subject to $X \in \mathcal{K}$. By  Corollary~\ref{maincor}, we conclude that  the opinion minimization problem is convex, and thus solvable in polynomial time.
\end{proof}

\subsection{Unbudgeted Opinion Maximization}

We also prove that the opinion maximization problem can also be solved in polynomial time using machinery that we have already developed. We state the problem as:

\begin{align*}
\begin{array}{ll@{}ll}
\text{maximize}  & \displaystyle \vec{1}^T   (X-(X-I)P)^{-1} s &  & \\
\text{subject to}        &X \succeq \frac{1}{u} I  & \\  
                 & X \preceq \frac{1}{\ell} I& \\ 
                 & X \text{~~diagonal} & \\ 
\end{array} 
 \end{align*}

\begin{theorem}
We can solve the opinion maximization problem by invoking our opinion minimization algorithm using as our vector of initial opinions $1-s$.
\end{theorem}

\begin{proof}
By Lemma~\ref{psd} and since $P\vec{1}=\vec{1}$,  for any $X \in \mathcal{K}$, 

\begin{align*}
&(X-(X-I)P) \vec{1} = X \vec{1} - (X-I) \vec{1} =  I \vec{1}  =\vec{1}\\
& \Leftrightarrow (X-(X-I)P)^{-1} \vec{1}  = \vec{1}. 
\end{align*}

\noindent Therefore, the opinion minimization objective that uses 
as an initial opinions the vector $\vec{1}-s$ becomes

\begin{align*}
\vec{1}^T   (X-(X-I)P)^{-1}  (\vec{1}-s) &=  \vec{1}^T   (X-(X-I)P)^{-1}  \vec{1} - \\ 
\vec{1}^T (X-(X-I)P)^{-1} s &= n - \vec{1}^T   (X-(X-I)P)^{-1} s. 
\end{align*}

\noindent Therefore, minimizing this objective is equivalent to maximizing 
$ \vec{1}^T * (X-(X-I)P)^{-1} s$ subject to $X \in \mathcal{K}$. 
\end{proof} 

We close this section by answering the following question: how much
can our intervention increase the total sum of opinions in a network?
As we have seen, Example \ref{eg:motivation} shows a way of modifying
resistance parameters to create an unbounded increase in the sum of 
expressed opinions.  Thus, the optimal modification to the resistance
parameters can produce arbitrarily large increases to this sum.



\section{Budgeted Opinion Optimization}
\label{sec:budgeted}
We now consider the setting where there is a constraint on the size of the target-set. That is, we want to identify a set $T \subseteq V$ of size $k$ 
such that changing the resistance parameters of agents in $T$ optimally maximizes (resp. minimizes) the sum of expressed opinions in equilibrium. 
Recall from Section \ref{sec:model} that we denote by $f(s, \alpha)$ the sum of expressed opinions under the given vector of innate values and resistance parameters. We denote by $f(s, \alpha, T)$ the sum of expressed opinions where we can set the $\alpha_i$ of the agents $i \in T$ optimally. We simply refer to this by $f(T)$ when there is no ambiguity. 

The objective function for finding the optimal target-set for opinion maximization is therefore  $\max_{|T| = k}f(T)$, 
and, respectively, $\min_{|T| = k} f(T)$ for opinion minimization.
Note, the objective function is monotone in the target-set since we
can always set the choice of $\alpha_i$ to be the original $\alpha_i$
for any node $i \in T$.
But, can we use submodular optimization tools for our budgeted opinion
optimization problems (Problems~\ref{probmax}, and \ref{probmin})? We
show that, unlike in the case in other standard opinion optimization or
influence maximization settings, neither submodularity nor
supermodularity holds. We prove this result for the case of opinion
maximization, but similar claims can be made for the minimization
problem too by setting $s = 1-s$.

\begin{eg}
Consider the complete graph on three nodes where $ s  = (1, 0, 0)$ and $\alpha = (0.1, 0.1, 0.1)$. The total sum of opinions in equilibrium is $1$. Suppose we are allowed to set resistance values in $[0.001, 1]$. Given target-sets $T' \subseteq T$ and $x \notin T'$, submodularity would give that, 
$$f(T \cup \{x \}) - f(T) \leq f(T' \cup \{x\}) - f(T').$$
Set $T = \{1, 2\}$,  $T' =\{1\}$, and $ x = \{3\}$. For each of these
sets, the optimal solution would set $\alpha_1 = 1$, $\alpha_2 =
0.001$ and $\alpha_3 = 0.001$. Submodularity would require

$$
f(1, 2, 3) - f(1, 2) \leq f(1, 3)  - f(1),
$$

but with these numbers, the left-hand side of the inequality is
0.191 and the right-hand is 0.168; 
this submodularity does not hold.

To show that supermodularity does not hold,  
set $T = \{1\}$, $T'  = \emptyset$, and $x = \{3\}$. 
Then, supermodularity would require that

$$
f(3) - f(\emptyset)  \leq f(1, 3)  - f(1),
$$

the left-hand side of the inequality is 1.493 
and the right-hand side is 0.168.

\noindent Hence, our objective function is neither submodular nor supermodular in the target set. 

\end{eg}

Highlighting a difference with the unbudgeted version, we show that
Problems~\ref{probmax}, and \ref{probmin} are \NPhard by a reduction
from the vertex cover problem on $d$-regular graphs.
To prove this, we adapt a construction for a different
opinion maximization problem given by 
\cite{gionis2013opinion} in the context of their model. 
We only present the proof for the opinion maximization problem;
as the example above, the proof can easily be adapted for opinion minimization
as well.

\begin{theorem}
The budgeted opinion maximization (and similarily opinion minimization) problem is \NPhard. 
\end{theorem}

\begin{proof}
Given a $d$-regular graph $G=(V, E)$ and an integer $K$, the vertex
cover problem asks whether there exists a set of nodes $S$ of size at
most $K$ such that $S$ is a vertex cover. An instance of the decision
version of opinion maximization consists of a graph $G'$, where each
agent $i$ has an internal opinion $s_i$ and resistance parameter
$\alpha_i$, an integer $k$ and a threshold $\theta$. The solution to
this decision version is ``yes" if and only if there exists a set $T$
of size at most $k$
whose resistance values we can optimally set such that $\ell \leq
\alpha \leq u $ and $f(T) \geq \theta$. We first
prove this hardness result by letting the $\alpha_i$ of the target-set
be anywhere in $[0, 1]$. We then prove that the reduction continues to
hold for the case where $\alpha_i \in (0, 1]$.

Suppose we are given an instance of the vertex cover problem for
regular graphs. We construct an instance of the decision version of
the opinion maximization problem as follows: set $G' = (V \cup V', E
\cup E')$ where $V'$ consists of $2 \sqrt{d}$ nodes for each node $i
\in V$ (called duplicate nodes). The set $E'$ is the edges of the form
$(i, i')$ joining nodes in $V$ to their corresponding duplicate nodes.
We set $s_i = 1$ and $\alpha_i = 0$ for all $i \in V$. We also set
$s_{i'} = 0$ for all $i' \in V'$. For the set of duplicate nodes, we
set the resistance value of half (that is, $\sqrt{d}$ nodes) to be $1$
and the other half to be $0$. We call the former the stubborn
duplicate nodes, and the latter the flexible duplicate nodes. We set
$k = K$ and $\theta = (\sqrt{d} + 1)k + \sqrt{d} (n -k) \frac{d }{d +
\sqrt{d}}$. This gives us a decision version of each of the opinion
maximization problem. 

We now want to show that a set $T$ of at most $k$ nodes in $V$ is a
vertex cover of $G$ if and only if we can set the resistance values of the
nodes in $T$ (as a subset of $G'$) to satisfy the decision version of
the opinion maximization problem.

Suppose that $T$ is a vertex cover of $G$. Then, it satisfies the
decision problem above by setting the $\alpha_i$ of the agents in $T$
to be $1$. Indeed, the agents $i \in T$ will have expressed opinions of $1$,
as will all the flexible duplicate nodes neighboring the agents in
$T$. All other nodes $j \in V$ will have all $d$ of their neighbors in
$V$ having opinion $1$ since $T$ is a vertex cover. Therefore, they
will have a final expressed opinion $\frac{d}{d + \sqrt{d}}$. Their
$\sqrt{d}$ flexible duplicate nodes will take on this value as well,
while their stubborn duplicate nodes will maintain their innate
opinion of $0$. Therefore, we get a final sum of $(\sqrt{d} +1 ) k +
\sqrt{d} (n - k) \frac{d}{d + \sqrt{d}}$. Therefore, $T$ corresponds
to a target-set whose resistance we can set optimally such that $f(T)
\geq \theta$.

To see the reverse direction, we want to show that if $T$ is not a
vertex cover of $G$, then we cannot set the resistance values of
agents in $T$ such that $f(T) \geq \theta$. If $T$ is not a vertex
cover, then there exists an edge $(i, j) \in E$ such that $i, j \notin
T$. Therefore, the expressed opinion of both $i$ and $j$ will be at
most $\frac{d-1}{d + \sqrt{d}}$. All the agents in $T$ will maintain
their high opinion of $1$, and all other nodes $\ell \neq i, j$ will
have an expressed opinion of at most $\frac{d}{d +  \sqrt{d}}$, by the
same argument as above. Taking into account the duplicate nodes, we
have, $$f(T) \leq (\sqrt{d}+1)k + \sqrt{d} (n - k - 2) \frac{d}{d
+ \sqrt{d}} + 2\sqrt{d}\frac{d-1}{d + \sqrt{d}} < \theta.$$
Let $\theta'$ denote the middle term in this chain of inequalities;
we have thus shown that if $G$ does not have a $k$-node vertex cover,
then for any set $T$ of at most $k$ nodes, we have 
$f(T) \leq \theta' < \theta$.

What remains to show is there exists an optimal target-set that consists of only nodes in $V$ and not in $V'$, where we assume that this target-set is at most of size $k \leq n$ for $n = |V|$. Denote the optimal target-set in by $T^*$. Suppose there is a node $i' \in T$ which is a node in $V'$. Let $i'$ be adjacent to a node $i \in V$. If $i'$ is a flexible duplicate node, then it is clear that the expressed opinions of both $i'$ and $i$ will be $1$, if we include $i$ in the target-set instead. (Or simply remove $i'$ from the target-set if $i$ is already in $T$.)
Suppose $i'$ is a stubborn duplicate node. First, assume that $i \in T$. Consider $T = T \backslash\{i'\} \cup \{ j \}$, where $j \notin T$. The expressed opinion of $j$, and therefore of all its corresponding flexible duplicate nodes, with $T$ as the target-set is at most $\frac{d}{d + \sqrt{d}}$. With $j$ in the target-set, the expressed opinion goes to $1$. We therefore need to show that the loss in opinion from removing $i'$ from the target-set, is at most the gain by adding $j$ to the target-set. Indeed, it is straightforward to verify that $1  \leq \({1 - \frac{d}{d + \sqrt{d}} } \) \( \sqrt{d} + 1\) $.

\hide{
\begin{align*}
1  \leq \({1 - \frac{d}{d + \sqrt{d}} } \) \( \sqrt{d} + 1\) &\Rightarrow \frac{1}{\sqrt{d} + 1} \leq \frac{\sqrt{d}}{d+ \sqrt{d}}\\
& \Rightarrow d + \sqrt{d} \leq d + \sqrt{d}
\end{align*}
} 
Now, consider $i'$ is a stubborn duplicate node and $i \notin T$. We compare $T$ to $T' = T \backslash \{i'\} \cup \{i\}$. Suppose there are $h$ nodes adjacent to $i$ in $V'$ with resistance parameter $0$. (In particular, $h \geq \sqrt{d} + 1$.) The expressed opinion of $i$ with $T$ as the target-set is at most $\frac{d}{d + 2\sqrt{d} - h}$, and is shared by all $h$ of the duplicate nodes. Under $T'$, $i$ will have an expressed opinion of $1$, as will all $h$ of the nodes with $\alpha = 0$, except for node $i'$ which will now have resistance value of $1$. Thus, we need to show that  $$1 - \frac{d}{d + 2 \sqrt{d} - h} \leq h \bigg( 1 - \frac{d}{d + \sqrt{d} - h}  \bigg).$$ This  holds since $h \geq 1$. 

We now show how to modify the proof for the case in which
the stubbornness values belong to the interval $[\varepsilon,1]$
rather than $[0,1]$, for some $\varepsilon > 0$.
The key point is that when we view $f(T)$ as a function of the 
stubbornness parameters in the network $G'$ constructed
in the reduction, it is a continuous function.
Thus, if we let $\gamma = \theta - \theta' > 0$, we can choose
$\varepsilon > 0$ small enough that the following holds.
If we take all nodes $i$ in the previous reduction with $\alpha_i = 0$,
and we push their stubbornness values up to $\alpha_i = \varepsilon$, then
(i) if $T$ is a $k$-node vertex cover in $G$, we have 
$f(T) \geq \theta - \gamma/3$, and 
(ii) if $G$ does not have a $k$-node vertex cover, for any 
set $T$ of at most $k$ nodes, we have $f(T) \leq \theta' + \gamma/3$. 
Since $\theta - \gamma/2$ lies between these two bounds, it follows
that with stubbornness values in the interval $[\varepsilon,1]$,
there is a $k$-node vertex cover $T$ in the original instance if and only if
there is a set $T$ in the reduction to opinion maximization with
$f(T) \geq \theta - \gamma/2$.
This completes the proof with stubbornness values in the 
interval $[\varepsilon,1]$. 
\end{proof}

In light of this hardness result, and inspired by Theorem \ref{thm:esv}, we propose a greedy heuristic for the budgeted opinion optimization problems. The version shown below is for opinion maximization, but can be converted to the opinion minimization problem by setting $s = 1-s$. In each round, the heuristic chooses the node that gives the largest marginal gain $f(T \cup \{ v  \})-f(T)$. The algorithm requires $O(n^3)$ time to compute $f(\emptyset)$ using standard Gauss-Jordan elimination for matrix inversion, and then it performs $O(nk)$ loops, that each requires $O(n^2)$ time.  To see why $O(n^2)$ time suffices instead of the straightforward $O(n^3)$ time of matrix inversion, suppose that we know $(I-(I-A)P)^{-1}$ and we wish to change the stubbornness of node $i$ from $\alpha_i$ to be $1$. The new stubbornness matrix is $A'= A + x_ix_i^T$, where $x_i \in \field{R}^n$ has $x_i(j)=0$ for all $j \neq i$, and $x_i(i) = \sqrt{ 1-\alpha_i}$. Instead of computing the inverse $(I-(I-A')P)^{-1}$, we can use the Sherman-Morrison formula for the rank-1 update $(I-(I-A)P + x_i (x_i^TP))^{-1}$, which leads to a run-time of $O(kn^3)$ instead of $O(kn^4)$.

\begin{algorithm}[!ht]
\label{greedy1} 
\floatname{algorithm}{Greedy Heuristic}  
\renewcommand{\thealgorithm}{}  
\caption{}  
\begin{algorithmic}[1] 
\REQUIRE $k \geq 1$ 
\STATE  $S \leftarrow \emptyset$
\STATE Compute $f(\emptyset) = (I - (I-A)P)^{-1}As$
\FOR{$i \leftarrow 1$ to $k$ }
\STATE $max \leftarrow -\infty$ 
	\FOR {$ v \in V \backslash T$ }
	\STATE Compute $f(T \cup \{ v  \})$
	 \IF{ $f(T \cup \{ v  \})-f(T) \geq max$}
	 \STATE  winner$\leftarrow v$ 
	 \STATE $max \leftarrow f(T \cup \{  v \}) -f(T) $
	 \ENDIF
	\ENDFOR
\STATE $T \leftarrow T \cup \text{winner}$
\ENDFOR 
\STATE Return $T$. 
\end{algorithmic}
\end{algorithm}

In the next section, we implement this heuristic and compare its performance against some natural baselines on real and synthetic values on real-world networks.

\section{Experimental results}
\label{sec:exp}
\spara{Experimental setup.} Table~\ref{tab:datasets} shows the datasets we used in our experiments.  We use several datasets as sources of real network topologies. The Twitter dataset is associated with real-world opinions obtained as described by De et al. \cite{de2014learning}.  For our purposes, we take the first expressed opinion by each agent to be the innate opinion, and we normalize these values, which fall in $[-1, 1]$ in the original dataset, to fall in $[0,1]$. The sum of the innate opinions is $347.042$ with a mean of $0.633$, while the sum of the final expressed opinions observed in the network during this period is $367.149$ with a mean of $0.670$, suggesting that the sentiment shifted in a positive manner.  For a detailed description of the dataset collection, see \cite{de2014learning}. 

\begin{table}[!ht]
\begin{center}
\begin{tabular}{|l|c|c|c|c|} \hline
Name            & $n$  & $m$  & Description & Source\\ \hline \hline
Monks     &  18 & 41 &   Social & \cite{monk1}  \\ \hline 
Karate     & 33 & 78  &   Social & \cite{karate} \\ \hline 
LesMiserables     & 77   &  254 & Co-occurance & \cite{les}  \\ \hline 
NetScience   &  379 & 914 &  Collaboration  & \cite{netsci} \\\hline  
Email Network & 986 & 16064 & Social & \cite{email1,email2} \\ \hline \hline
Twitter & 548 & 3638 & Social Media & \cite{de2014learning} \\ \hline \hline 
\end{tabular}
\end{center}
\caption{\label{tab:datasets} Description of datasets.}
\end{table}

For the remaining networks, we set the innate opinions uniformly at random in $[0,1]$ for the budgeted experiments and uniformly at random as well as power law for the unbudgeted experiments. We set the slope for the power law distribution to be $2$. For all graphs, we  select each resistance parameter $\alpha_i$ uniformly at random in $[0.001, 1]$.

For the budgeted version of our problem, we present the performance of the  heuristics for the three largest networks (i.e., Twitter, Email Network, and NetScience). The results are qualitatively similar for the rest of the networks. We compare the performance of our greedy heuristic with two  baselines that output a set of $k$ nodes whose resistance parameter is set to $1$. Baseline I chooses  $k$ agents with the highest innate opinions (ties broken arbitrarily).  Baseline II assigns to each node $i$ a score $\text{score}(i)= \frac{\text{deg}(i)}{2m} \times \frac{s_i}{\sum_{j \in N(i)}s_j}$ and chooses the $k$ agents with the highest score. The intuition for this baseline is that we would like to pick a node that has a high centrality and relatively high opinion compared to their neighbors and fix this agent to have an $\alpha$ of 1. We consider budget values $k$ in $\{1, 2, \cdots, 100\}$. We compare the sum of expressed opinions without intervention as well as the outputs of the greedy heuristic and the two baselines. All simulations run on a laptop with 1.7 GHz Intel Core i7 processor and 8GB of main memory.\footnote{ The code for the budgeted and unbudgeted experiments can be found in \url{https://github.com/redabebe/opinion-optimization} and \url{https://github.com/tsourolampis/opdyn-social-influence}, respectively.}

 \begin{algorithm}[H]
\label{greedy1} 
\floatname{algorithm}{Baseline I}  
\renewcommand{\thealgorithm}{}  
\caption{}  
\begin{algorithmic}[1] 
\REQUIRE $k \geq 1$ 
\STATE Let  $s_1 \geq s_2 \geq \ldots \geq s_n$. 
\STATE Return $T \leftarrow  \{1,\ldots,k\}$ 
\end{algorithmic}
\end{algorithm}

 \begin{algorithm}[H]
\label{greedy1} 
\floatname{algorithm}{Baseline II}  
\renewcommand{\thealgorithm}{}  
\caption{}  
\begin{algorithmic}[1] 
\REQUIRE $k \geq 1$ 
\STATE Let $\text{score}(i) \leftarrow \frac{\text{deg}(i)}{2m} \times \frac{s_i}{\sum_{j \in N(i)}s_j}$  for all $i \in V$.
\STATE Sort the nodes such that $\text{score}(1) \geq \text{score}(2) \geq \ldots \geq \text{score}(n)$ 
\STATE Return $T \leftarrow  \{1,\ldots,k\}$ 
\end{algorithmic}
\end{algorithm}

\spara{Experimental findings.} Figures~\ref{fig:budgeted}(a), (b), (c) show the results obtained by  our greedy heuristic, and the two baseline methods for the Twitter, Email, and NetScience networks respectively.  All methods achieve a significant improvement compared to having no interventions even for small values of $k$. Recall that for the Twitter
network, we use estimates of real opinions for the agents. Homophily by opinions is
observed and impacts the relatively higher performance of these baselines compared to
the Email and NetScience networks where the opinion values are randomly generated.
Baseline I and Baseline II show similar performance for the Twitter and NetScience
networks. One underlying explanation for the relative under-performance of Baseline
II compared to Baseline I in the Email network is that this network has skewed degree
distribution which might inflate the scores of higher degree nodes. We note that the run 
times for the Twitter, Email, and NetScience networks are 12.0, 3.4, 
and 111.5 seconds per k, respectively.

\begin{figure*}[!ht]
\centering
\includegraphics[scale=0.35]{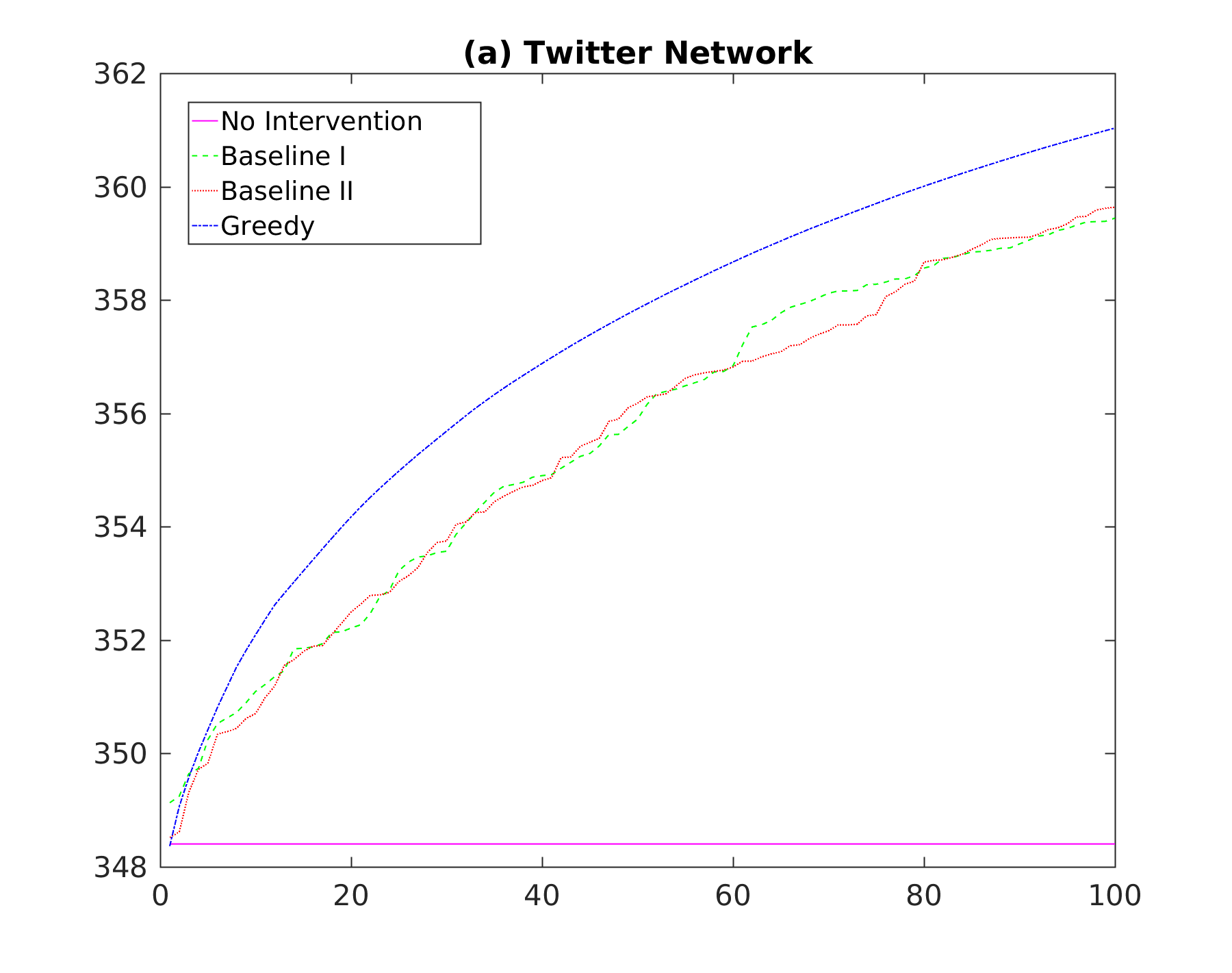} 
\includegraphics[scale=0.175]{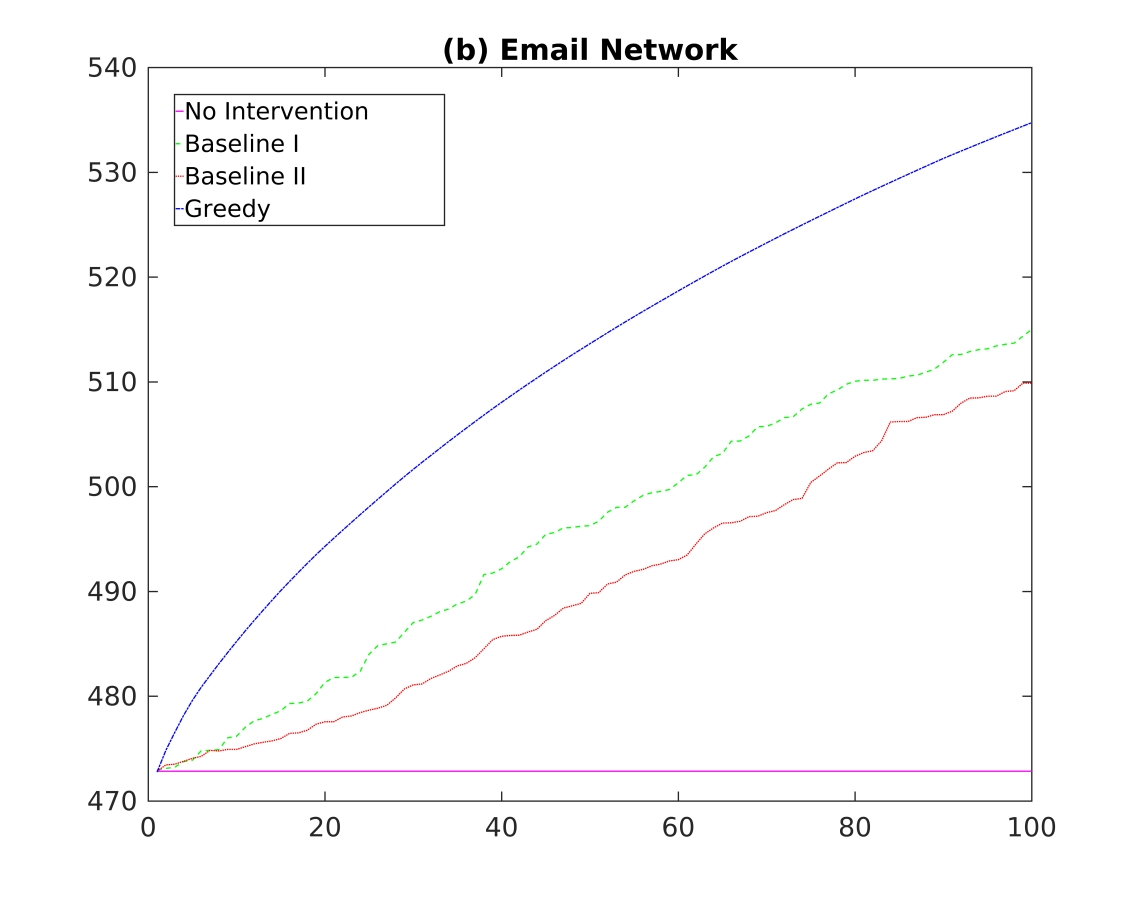}  
\includegraphics[scale=0.35]{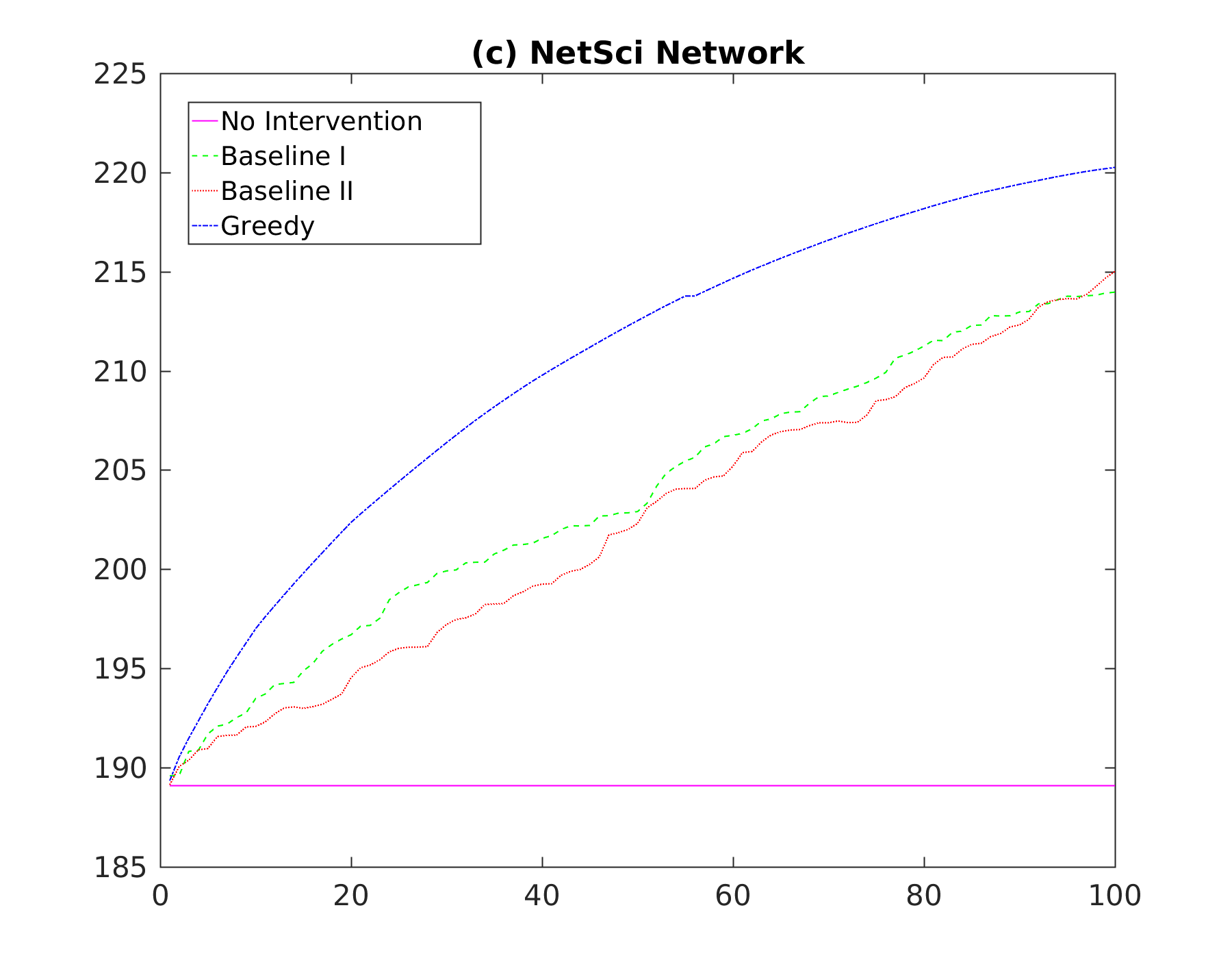} \\  
\caption{\label{fig:budgeted} Total sum of opinions at equilibrium vs. budget $k$ for the budgeted opinion maximization problem using no intervention, Baseline I, Baseline II, and our greedy heuristic for (a) Twitter, (b) Email,  and (c) NetScience networks respectively. 
}
\end{figure*}

Table~\ref{tab:synthetic} shows our experimental findings for the unbudgeted opinion minimization, and maximization  problems. We report the total sum $\vec{1}^Ts$ of innate opinions, the total sum $\vec{1}^Tz$ of opinions at equilibrium (no intervention), the total sum $\vec{1}^Tz_{\text{min}}$ of opinions produced by our minimization algorithm, and the total sum $\vec{1}^Tz_{\text{max}}$ of opinions produced by our  maximization algorithm. The resistance parameters for each of the networks are chosen uniformly at random from $[0,1]$. We then let $\ell = 0.001$ and $u = 1$, which serve as an upper and lower bound for the resistance parameters. Note that the reported values are the average of 5 trials. In addition to the synthetically generated opinion values reported above, we also run the opinion optimization algorithms on the Twitter dataset using the opinions from \cite{de2014learning} and get a sum of opinions $117.31$ for the opinion minimization problem and $363.87$ for the opinion maximization problem. The main observation for here is that the effect of our intervention can be significant for both the opinion maximization and opinion minimization problems. Indeed, the sum of opinions in the case of the former is over an order of magnitude larger than the case of the latter for many of these networks.

\begin{table*}[!ht]
\begin{center}
\begin{tabular}{|c|c|c|c|c|c|c|c|c|}\hline 
$G$        &  \multicolumn{4}{|c|}{$s\sim U[0,1]$}  &  \multicolumn{4}{|c|}{$s \sim$Power Law($\alpha$)} \\ \hline
     
			 &  \multicolumn{1}{|c|}{$\vec{1}^Ts$} &  
			 \multicolumn{1}{|c|}{$\vec{1}^Tz$} &  
			 \multicolumn{1}{|c|}{$\vec{1}^Tz_{\text{min}}$} &  
 \multicolumn{1}{|c|}{$\vec{1}^Tz_{\text{max}}$} &  	
	 \multicolumn{1}{|c|}{$\vec{1}^Ts$} &  
			 \multicolumn{1}{|c|}{$\vec{1}^Tz$} &  
			 \multicolumn{1}{|c|}{$\vec{1}^Tz_{\text{min}}$} &
			  \multicolumn{1}{|c|}{$\vec{1}^Tz_{\text{max}}$}     \\ \hline 		
			 Monks         & 9.34  & 8.12 & 1.22 &  16.38   & 11.88 & 11.29 & 5.62   &  16.68 \\ \hline 	
			 Karate  & 16.70  & 17.83 & 1.97 &  32.47 &  22.29 & 21.07 & 5.27 & 33.56   \\ \hline 	
			 LesMis     & 38.81 & 40.5   & 4.21 & 74.00  & 48.68 & 52.11 & 8.44  & 75.84  \\ \hline
			 NetScience               & 188.39  & 191.79 & 17.19 & 319.59 &  257.58 & 249.88 & 64.43 & 340. 54 \\ \hline 	
			 Email               & 414.62  & 440.96 & 388.39 & 495.44 &  283.62 & 565.86 & 91.18 & 610.72 \\ \hline 	
			 Twitter               & 289.29  & 260.96 & 241.86 & 270.21 &  193.09 & 359.85 &  62.00 & 379.00 \\ \hline 	
\end{tabular}       
\end{center}
\caption{Results for unbudgeted algorithm for two different distributions (uniform, and power law with slope $\alpha$ equal to 2) of innate opinions $s$:   the total sum $\vec{1}^Ts$ of innate opinions , the total sum $\vec{1}^Tz$ of opinions at equilibrium (no intervention), the total sum $\vec{1}^Tz_{\text{min}}$ of opinions produced by our minimization algorithm, and the total sum $\vec{1}^Tz_{\text{max}}$ of opinions produced by our  maximization algorithm.}
\label{tab:synthetic}
\end{table*}

\section{Related Work}
\label{sec:related}
To our knowledge, we are the first to consider an optimization framework based on opinion dynamics with varying susceptibility to persuasion. In the following we review briefly some work that lies close to ours. 

\spara{Susceptibility to Persuasion.} Asch's conformity experiments
are perhaps the most famous study on the impact of agents'
susceptibility to change their opinions \cite{asch}. This study shows
how agents have different propensities for conforming with others.
These propensities are modeled in our context by the set of parameters
$\alpha$. Since the work of Asch, there have been various theories on
peoples' susceptibility to persuasion and how these can be affected. A
notable example is Cialdini's Six Principles of Persuasion, which
highlight reciprocity, commitment and consistency, social proof,
authority, liking, and scarcity, as key principles which can be
utilized to alter peoples' susceptibility to persuasion
\cite{cialdini1,cialdini2}. This framework, and others, have been
discussed in the context of altering susceptibility to persuasion in a
variety of contexts.  Crowley and Hoyer \cite{crowley}, and McGuire
\cite{mcguire} discuss the `optimal arousal theory', i.e., how novel
stimuli can be  utilized for persuasion when discussing arguments.

\spara{Opinion Dynamics Models.}  Opinion dynamics model social learning processes. DeGroot introduced a continuous opinion dynamics model in his seminal work on consensus formation \cite{de}. A set of $n$ individuals in society start with initial opinions on a subject. Individual opinions are updated using the average of the neighborhood of a fixed social network. Friedkin and Johnsen \cite{fj} extended the DeGroot model to include both disagreement and consensus by mixing each individual's \emph{innate belief} with some weight into the averaging process. This has inspired a lot of follow up work, including \cite{ao,asu,bko,mod,gs,gtt}.

\spara{Optimization and Opinion Dynamics.}  Bindel et al. use the
Friedkin-Johnsen model as a framework for understanding  the price of
anarchy in society when individuals selfishly update their opinions in
order to minimize the stress they experience \cite{bko}. They also
consider network design questions: given a budget of $k$ edges, and a
node $u$, how should we add those $k$ edges to $u$ to optimize an
objective related to the stress?   Gionis et al. \cite{gtt} use the
same model to identify a set of target nodes whose innate opinions can
be modified to optimize the sum of expressed opinions.
Finally, in work concurrent with ours, 
Musco et al. adopt the same model to understand
which graph topologies minimize the sum of disagreement and
polarization \cite{musco2017minimizing} .

\spara{Inferring opinions and conformity parameters.} While the
expressed opinion of an agent is readily observable in a social network,
both the agent's innate opinion and conformity parameter are hidden, and
this leads to the question of inferring them. Such
inference problems have been  studied by Das et al. \cite{role,deb}.
Specifically, Das et al. give a near-optimal sampling algorithm for
estimating the true average innate opinion of the social network and
justify the algorithm both analytically and experimentally \cite{deb}.
Das et al.  view the problem of susceptibility parameter estimation as
a problem in constrained optimization and give efficient algorithms,
which they validate on real-world data \cite{role}. There are also
several studies that perform experiments to study what phenomena
influence opinion formation and how well these are captured by
existing models \cite{mod,socinf}.

\spara{Social Influence.} The problem we focus on is related to
influence maximization, a question that draws on work of
Domingos and Richardson \cite{dr,rd} and Kempe et al.  \cite{kkt}.
Recent work by Abebe et al.  considers the problem of influence
maximization where agents who reject a behavior contribute
to a negative reputational effect against it~\cite{aak}. 
There is also a rich line of
empirical work at the intersection of opinion dynamics and influence
maximization including \cite{hb,turf,berger2007consumers} that argues
that agents adopt products to boost their status.  

\section{Conclusion}
\label{sec:concl}
Inspired by a long line of empirical work on modifying agents'
propensity for changing their opinions,  we initiate the study of
social influence by considering interventions at the level of
susceptibility to persuasion. For this purpose, we adopt a popular
opinion dynamics model, and seek to optimize the total sum of
opinions at equilibrium in a social network. We consider
both opinion maximization and minimization, and focus on two
variations: a budgeted and an unbudgeted version. We prove that 
the unbudgeted problem is solvable in polynomial time (in both
its maximization and minimization version). 
For the budgeted version, on the other hand,
in contrast to a number of other opinion or influence maximization formalisms,
the objective function is neither submodular nor supermodular. 
Computationally, we find that it is \NPhard.
In light of this, we introduce a greedy heuristic for the budgeted
setting, which we show has desirable performance on various real and
synthetic datasets.

There are various open directions suggested by this framework. The
first is whether we can provide any approximation algorithms for the
budgeted setting. Through simulations, we also note that while not
giving optimal results, local search algorithms yield satisfactory
results in finding target-sets in the unbudgeted setting. Exploring
this connection might yield insights into the opinion formation
process. The interface of influence maximization and opinion
maximization is also under-explored, and yet another interesting line
of work is to note equivalences for models under the two frameworks,
and when they depart from one another. Many of the questions explored
in Section \ref{sec:related} can also be posed in this setting
including budgeted edge-addition and removal with susceptibility to
persuasion, and understanding the influence of opinions in speed of
convergence to equilibrium. Finally, we can also pose the question of
opinion optimization when there is a budget on the total sum of changes
we can make to the resistance values of nodes,
rather than on the number of agents who can be in a target set.

\section*{Acknowledgements} 
Rediet Abebe was supported in part by a Google scholarship, a Facebook scholarship,
and a Simons Investigator Award, and Jon Kleinberg was supported in part by a Simons
Investigator Award, an ARO MURI grant, a Google Research Grant, and a Facebook
Faculty Research Grant. 

We thank the authors of \cite{de2014learning} for sharing the Twitter dataset in Section \ref{sec:exp}.

Charalampos Tsourakakis would like to thank his newborn son Eftychios for the happiness he brought to his family.

\end{document}